\documentclass[11pt]{article}

\input{conf}

\title{Tight Approximation 
Results for Matroid Optimization with a Linear Constraint}

\author{Ilan Doron-Arad\thanks{Math Department, MIT, Cambridge, Massachusetts, USA. \texttt{ilanda@mit.edu}} \and
Hadas Shachnai\thanks{Computer Science Department, Technion, Haifa, Israel. \texttt{hadas@cs.technion.ac.il}} \and 
Gilad Shmerler\thanks{Computer Science Department, Technion, Haifa, Israel. \texttt{shmerler@campus.technion.ac.il}}
}
\date{}

\begin{document}
\maketitle

\begin{abstract}
We study the following class of \emph{matroid optimization problems with a linear constraint} ($\mathcal{P}$-MOL). Given a matroid $\mathcal{M}=(E,\mathcal{I})$, two nonnegative weight functions $v,w:E\to\mathbb{R}_{\ge 0}$, and a threshold $L\in\mathbb{R}_{\ge 0}$, find $\textsf{opt}\, v(S)$ where $S$ is either an independent set or a base of $\mathcal{M}$ satisfying a budget-type constraint:  $w(S)\le L$ or $w(S)\ge L$, and $\textsf{opt}\in\{\min,\max\}$.
$\mathcal{P}$-MOL provides a unified representation for a broad family of NP-hard optimization problems, including
{\sc budgeted matroid independent set}, {\sc constrained minimum-basis}, 
and {\sc knapsack-cover} variants with a matroid constraint. Also, it naturally extends to multiple matroid constraints. In particular, we consider the {\sc matroid intersection cover} (MIC) problem, where feasibility is defined by the common
independent sets of two matroids and one seeks minimum $v(S)$ subject to $w(S)\ge L$.

Our main result is a unified {\em efficient polynomial-time approximation scheme} (EPTAS) for all nontrivial \pmol\ variants, obtained by generalizing a technique of Hassin and Levin (SIAM J.\ Comput., 2004) for solving the {\sc constrained minimum spanning tree} problem.
Specifically, for any fixed $\varepsilon > 0$, we present an algorithm running in time $|E|^{O(1)} \cdot (1/\varepsilon^2)^{O(1/\varepsilon)}$ that outputs a feasible solution $S$ whose value is at most $(1+\varepsilon)\textsf{OPT}$ for minimization variants and at least $(1-\varepsilon)\textsf{OPT}$ for maximization variants.
This resolves the complexity status of all members of $\mathcal{P}$-MOL,
as none of these problems admits a {\em fully} polynomial-time approximation scheme (Doron-Arad, Kulik and Shachnai, ICALP'24).
Finally, we separate the $\mathcal{P}$-MOL family from its extension to matroid intersection. We show that an EPTAS is unlikely to exist for the matroid intersection variant of $\mathcal{P}$-MOL under a covering constraint, whereas an EPTAS is known to exist under a budget constraint. This highlights a qualitative difference between these two types of linear constraints that does not arise in the single-matroid setting.
\end{abstract}

\thispagestyle{empty}
\clearpage

\setcounter{page}{1}

\section{Introduction}

Matroids offer a unifying abstraction for ``independence'' in combinatorial structures: they capture linear
independence in vector spaces, forests in graphs, and many other systems constrained by an exchange principle. 
A {\em matroid} is a set system $(E, \cI)$, where $E$ is a finite set and $\cI \subseteq 2^E$ are the {\em independent sets (IS)} such that $(i)$ $\emptyset \in \cI$, $(ii)$ for all $A \in \cI$ and $B \subseteq A$ it holds that $B \in \cI$, and $(iii)$ for all $A,B \in \cI$ where $|A| > |B|$, there is $e \in A \setminus B$ such that $B \cup \{e\} \in \cI$.\footnote{Properties $(ii)$ and $(iii)$ are known, respectively,  as {\em hereditary property}, and {\em exchange property}.}
A broad range of classical problems can be cast as \emph{matroid optimization}: one seeks a minimum-cost or maximum-value independent set or basis of a matroid.
Canonical examples include minimum spanning trees, maximum-weight forests, and a variety of assignment and
scheduling problems.

In many applications, however, matroid feasibility alone is insufficient. Solutions must also satisfy an
additional \emph{budget-type} constraint, such as an upper bound on total length or cost, or a lower bound on
total coverage. Adding even a single linear constraint of this form typically turns an otherwise polynomial-time
matroid optimization problem into an NP-hard one, and has motivated a long line of work on approximation algorithms.

\subsection{The \pmol Class and Our Framework}
We start by defining the class of problems studied in this paper.
Let $\scm = (E,\sci)$ be a matroid, and let $v,w : E \to \drz$ be two nonnegative functions. The value $v_e$ contributes to the main objective, and $w_e$ contributes to a budget constraint. For the matroid $\scm = (E,\sci)$ we define $\IS(\scm) := \sci$ and $\bases(\scm) := \{\, S \in \sci : |S| = \rank(\scm) \,\}$, where $\rank(\scm) := \max\{\,|T| : T \in \sci\,\}$ is the rank of $\scm$, i.e., the maximum cardinality of an independent set. 

We are interested in three design choices:
\begin{itemize}
  \item \textbf{Objective direction.}
        We may want to minimize or maximize $v(S)$.\footnote{For any vector $p = (p_e)_{e \in E}$ and any $S \subseteq E$ we use the shorthand $p(S) := \sum_{e \in S} p_e$.}
  \item \textbf{Matroid family.}
        We may optimize over independent sets or bases of $\scm$.
  \item \textbf{Budget direction.}
        The budget constraint may take the form $w(S) \le L$ or $w(S) \ge L$.
\end{itemize}
We encode these choices by three parameters
\begin{equation*}
  \opt \in \{\min,\max\},
  \qquad
  \scf \in \{\IS, \bases\},
  \qquad
  \triangleright \in \{\le,\ge\}.
\end{equation*}
Given a threshold $L \ge 0$, we consider the optimization problem
\begin{equation}\label{eq:pmol-def}
  \opt\; v(S)
  \qquad\text{s.t.}\qquad
  S \in \scf,
  \quad
  w(S)\;\triangleright\; L.
\end{equation}
We refer to \eqref{eq:pmol-def} as a \pmol (Matroid Optimization with a Linear constraint) problem, with the descriptor $\scp = (\opt, \scf, \tr)$.
An instance of a \pmol problem is specified by a tuple $\sca_\scp = (E,\sci,v,w,L)$.

For a fixed descriptor $\scp$, we denote the optimal value of~\eqref{eq:pmol-def}
by $\OPT_\scp$, or simply $\OPT$ when the context is clear:
\[
  \OPT_\scp
  \;:=\;
  \opt\bigl\{\,v(S) : S \in \scf(\scm),\ w(S) \triangleright L\,\bigr\},
\]

This formulation captures, in a uniform way, several extensively studied problems. For example:
\begin{itemize}
  \item $(\max,\IS,\le)$ gives the {\sc budgeted matroid independent set} problem, generalizing $0/1$-{\sc knapsack}.
  \item $(\min,\bases,\le)$ gives {\sc constrained minimum-basis} problems, including {\sc constrained minimum spanning tree} when $\scm$ is a graphic matroid.
  \item $(\min,\IS,\ge)$ corresponds to {\sc knapsack-cover with a matroid}. 
\end{itemize}
Among the eight possible choices of $(\opt,\scf,\triangleright)$, only $(\min,\IS,\le)$ is trivial (its optimum is
$S=\emptyset$); we focus on the remaining seven nontrivial variants.

The \pmol viewpoint also naturally extends beyond a single matroid. If one lets $\scf$ denote the intersection of two matroids, then \eqref{eq:pmol-def} becomes matroid-intersection optimization with a linear constraint. A notable example is the {\sc matroid intersection cover (MIC)} problem: given matroids $\scm_1=(E,\sci_1)$ and $\scm_2=(E,\sci_2)$ and weights $v,w:E\to\dr_{\ge 0}$, find a set $S\in\sci_1\cap\sci_2$ with $w(S)\ge L$ minimizing $v(S)$.
In our terminology, MIC aligns with a $(\min,\mathsf{intersection},\ge)$-MOL formulation, with feasibility governed by $\sci=\sci_1\cap\sci_2$.

\subsection{Our Contribution}
We present (in \Cref{sec:eptas}) a unified approximation framework for \pmol problems by generalizing a technique of Hassin and Levin~\cite{hassin2004efficient}. 
This allows to obtain an EPTAS for all non-trivial variants of $\scp$.
Specifically, for every fixed $\eps>0$ there exists an algorithm which, given an instance $\sca_\scp$, outputs a feasible solution $S$ satisfying
\[
  \opt = \min \;\Longrightarrow\; v(S) \le (1+\varepsilon)\,\OPT_\scp,
  \qquad
  \opt = \max \;\Longrightarrow\; v(S) \ge (1-\varepsilon)\,\OPT_\scp.
\]
The running time is $h(1/\eps)\cdot \poly(|\sca_\scp|)$ for an arbitrary function $h$.
This resolves the complexity status of all members of $\scp$-MOL,
as none of these problems admits a {\em fully} polynomial-time approximation scheme (FPTAS)~\cite{DKS24}. Table~\ref{tab:pmol-summary} presents the best known results for the \pmol class. Results obtained in this paper are marked in boldface.

We also consider the natural extension of \pmol variants to matroid intersection. For 
{\sc budgeted matroid intersection}, an EPTAS was presented in~\cite{DKS23b}.
In contrast to our positive result for the \pmol class, where matroid optimization is over a {\em single} matroid constraint,
we show that a unified EPTAS is unlikely to exist for {\em all} \pmol variants, already when the single matroid constraint is replaced by a {\em matroid intersection} constraint. Specifically, we prove (in \Cref{sec:hardness}) that MIC does not admit an EPTAS unless $\mathrm{W[1]}=\mathrm{FPT}$.

\subsection{Techniques}

\label{sec:techniques}

\paragraph{A Unified EPTAS for \pmol\!.}
For {\sc Budgeted Matroid Independent Set}, as well as for its extension to {\sc Budgeted Matroid Intersection}, recent works~\cite{DKS23,DKS23b} give EPTASs based on the \emph{representative-set} methodology. Roughly, it computes a small candidate pool that is guaranteed to contain the high-value elements of some near-optimal solution, and then recovers the near-optimum by exhaustive enumeration over this pool. However, this approach inherently $(i)$ loses the profit of $O(1)$ elements, and $(ii)$ may drop $O(1)$ elements to restore feasibility. 
Thus, a solution output by the representative set approach may be a non-basis for the given matroid. It may also be infeasible w.r.t. a covering constraint. Hence, the EPTASs of~\cite{DKS23,DKS23b} cannot be extended to EPTASs for the {\em entire}
\pmol class, which must handle both $\IS(\scm)$ and $\bases(\scm)$, and all four choices of $(\opt,\triangleright)$.

Instead, we generalize (in \Cref{sec:eptas}) a technique of Hassin and Levin~\cite{hassin2004efficient} to obtain a unified EPTAS for all members of \pmol under a general matroid constraint. The following is a high-level overview of the algorithmic framework.

Given a constant-factor guarantee $\hat V$ for $\OPT$, we bucket elements according to their values $v_e$. Elements with $v_e \le \varepsilon \hat V$ are classified as \emph{small}, while the remaining elements are partitioned into $O(1/\varepsilon^2)$ narrow buckets. We note that any optimal solution contains only $O(1/\varepsilon)$ \emph{large} elements. This allows us to enumerate \emph{configuration vectors} $\vec{k}$ that specify the number of selected elements from each bucket.

For a fixed configuration $\vec{k}$, we enforce these counts by intersecting the matroid with a partition matroid, yielding a matroid-intersection family 
$\mathcal{F}_{\vec{k}}$. We then use a Lagrangian relaxation of the linear constraint $w(S)\triangleright L$ and compute $\lambda^\star_{\vec{k}}$, an optimal value of the multiplier $\lambda$.

The main technical challenge is to convert the Lagrangian optimum into a feasible solution. At $\lambda^\star_{\vec{k}}$, the family $\mathcal{O}_{\vec{k}}$ of Lagrangian-optimal sets is \emph{well connected}: any two of its members can be transformed into one another via a sequence of exchanges, each replacing at most one element per bucket. We compute sets $S^{-}, S^{+} \in \mathcal{O}_{\vec{k}}$ that respectively minimize and maximize the signed violation of $w(S)\triangleright L$, and hence lie on opposite sides of feasibility. Traversing the one-per-bucket exchange path from $S^{-}$ to $S^{+}$, we select the first feasible set $S^\star$.

Because elements within each bucket have nearly identical $v$-values, each exchange affects the objective value by at most the total \emph{bucket spread}—that is, the difference between the maximum and minimum $v$-value within a bucket. Consequently, the first feasible set $S^\star$ yields an objective value within an additive $O(\varepsilon \OPT)$ of the Lagrangian bound. Taking the best such $S^\star$ over all configurations yields an EPTAS.

We note that some variants in our descriptor-based formulation are related by standard reductions. Nevertheless, we keep the unified formulation for completeness and to make the scope of the algorithm explicit, since such reductions may not always preserve the approximation guarantees in a straightforward way. In particular, budget-type and covering-type constraints require careful treatment: the existence of an EPTAS for the budget setting under a matroid intersection constraint~\cite{DKS23b} does not, by itself, imply an analogous result for the covering setting. Our hardness result in \cref{thm:MIC} highlights this distinction. Therefore, we present the argument directly for the descriptor class $\scp=(\opt,\scf,\tr)$.

\paragraph{Hardness of MIC.}
We establish a parameterized lower bound for {\sc matroid intersection cover} (MIC). Via a reduction from $k$-Clique, we show that MIC admits no EPTAS unless $W[1]=\mathrm{FPT}$ (Theorem~\ref{thm:MIC}). The novelty of our result lies in the fact that this reduction rules out an \emph{EPTAS}, whereas a similar construction of~\cite{DKS24} is used only to exclude an \emph{FPTAS}. Our result complements recent applications of related reductions to $3$-{\sc matroid intersection}~\cite{doron2026you}.
Together, our findings demonstrate that for matroid intersection, different linear constraints (cover versus budget) lead to markedly different algorithmic behavior, in contrast to the case of a single matroid.


\begin{table}[ht]
\centering
\renewcommand{\arraystretch}{1.2}
\begin{NiceTabular}{@{} lccc c c @{}}
\toprule
\textbf{Problem Type} & \textbf{Opt} & $\mathcal{F}$ & $\triangleright$ & \textbf{Positive Result} & \textbf{Negative Result} \\
\midrule
{\sc Budgeted Matroid Independent Set} & $\max$ & $\text{IS}$ & $\le$ & EPTAS~\cite{DKS23} & No FPTAS \\
{\sc Knapsack-Cover with a Matroid} & $\min$ & $\text{IS}$ & $\ge$ & \textbf{EPTAS} & No FPTAS \\
\addlinespace
{\sc Budgeted Matroid Basis} & $\max$ & $\text{Bases}$ & $\le$ & \textbf{EPTAS} & No FPTAS \\
{\sc Constrained Minimum Basis} & $\min$ & $\text{Bases}$ & $\le$ & \textbf{EPTAS}  \tablefootnote{Previously, an EPTAS was known for the special case of {\sc Constrained Minimum Basis} in a {\em graphic} matroid~\cite{hassin2004efficient}.}
& No FPTAS \\
\addlinespace
{\sc Lower-Bounded Max Basis} & $\max$ & $\text{Bases}$ & $\ge$ & \textbf{EPTAS} & No FPTAS\\
{\sc Lower-Bounded Min Basis} & $\min$ & $\text{Bases}$ & $\ge$ & \textbf{EPTAS} & No FPTAS\\
\addlinespace
{\sc Lower-Bounded Indep. Set} & $\max$ & $\text{IS}$ & $\ge$ & \textbf{EPTAS} & No FPTAS\\
\bottomrule
\end{NiceTabular}
\caption{Our positive results for $\scp$-MOL problems and the matching negative results of~\cite{DKS24}.}
\label{tab:pmol-summary}
\end{table}



\subsection{Related Work}

\label{sec:related}

\paragraph{Knapsack and approximation schemes.}
The $0/1$-\textsc{Knapsack} problem is a canonical benchmark for approximation schemes. Early work already established a fully polynomial-time approximation scheme~\cite{La79}, and dynamic-programming–based FPTASs and their implications have since been extensively studied~\cite{CLRS22}. More recent progress has focused on fine-grained analyses and improved time–approximation tradeoffs~\cite{bringmann2021fine,deng2023approximating}.

\paragraph{Matroid and matching constraints with a single budget.}
Adding a matroid or matching feasibility constraint to $0/1$-\textsc{Knapsack} gives rise to the well-studied problems of {\sc Budgeted Matching} and {\sc Budgeted Matroid Intersection}. A PTAS for both variants was obtained by Berger et al.~\cite{BBGS11}. Subsequent work strengthened these results to EPTAS guarantees: for {\sc Budgeted Matroid Independent Set}~\cite{DKS23}, and for {\sc Budgeted Matching} and {\sc Budgeted Matroid Intersection} via representative sets~\cite{DKS23b}.

While an FPTAS is ruled out for general instances of {\sc Budgeted Matroid Independent Set}~\cite{DKS24}, FPTASs are known for several special cases, including {\sc Cardinality-Constrained Knapsack}~\cite{caprara2000approximation}, {\sc Multiple-Choice Knapsack}~\cite{kellerer2004multiple}, and laminar matroids~\cite{DKS23c,yang2025improved}. Additional variants of the problem have also been studied, for example in~\cite{doron2023budgeted,doron2025algorithm}.

$\scp$-MOL problems have also been studied in the covering setting (as opposed to a budget constraint), in conjunction with a matroid feasibility constraint. For {\sc Lower-Bounded Minimum Independent Set}, the best known result prior to this work was a PTAS, due to~\cite{chakaravarthy2013knapsack}. For {\sc constrained minimum spanning tree}, a PTAS was given in~\cite{ravi1996constrained}, followed by an EPTAS in~\cite{hassin2004efficient}.
Furthermore, a multi-criteria FPTAS is known for this problem~\cite{hong2004fully}.


\paragraph{Multiple budgets.}
When multiple budget constraints are present, PTASs are known for variants with either matroid intersection or matching constraints, using techniques such as dependent randomized rounding~\cite{CVZ11} and structural properties of matroid polytopes~\cite{GZ10}. On the other hand, even seemingly simple extensions face strong complexity barriers: for example, the {\sc Multidimensional Knapsack} problem with just two budget constraints already rules out an EPTAS under standard complexity assumptions~\cite{kulik2010there}.


\section{An EPTAS for the \pmol Class}
\label{sec:eptas}
Let $\sca_\scp = (E,\sci,v,w,L)$ be an instance of \pmol with descriptor $\scp = (\opt,\scf,\tr)$ and fix $\eps > 0$. Recall that $\scf$ denotes the underlying combinatorial family: either all independent sets or all bases of $\scm$. Our EPTAS combines a value-bucketing scheme with a Lagrangian relaxation applied to the intersection of $\scm$ with a partition matroid (defined below). We begin by describing the bucketing and configuration step. Let $n := |E|$ denote the ground-set size.

\subsection{Value Bucketing and Configuration Vectors}

We first obtain a constant-factor estimate of $\OPT$.

\begin{lemma}\label{lem:Vhat}
Given an instance $\sca_\scp$, one can find in $|E|^{O(1)}$ a feasible solution $\hat{S}$ with value $\hat{V} := v(\hat{S})$ such that
\[
  \frac{\OPT}{2} \le \hat{V} \leq \OPT \quad \text{if } \opt = \max,
  \qquad\text{and}\qquad
  \OPT \le \hat{V} \le 2\OPT \quad \text{if } \opt = \min.
\]
\end{lemma}

\begin{proof}
    We obtain $\hat V$ by iterating over all powers of $2$ in the domain $[0,v(E)]$. In one of these iterations we necessarily hit a value within a factor $2$ of $\OPT$, and we output this guess as $\hat V$. For simplicity of presentation, in the remainder we assume that we are given such a value $\hat V$ satisfying the stated $2$-approximation guarantees.

    The number of guesses is \(O(\log v(E))\). Since \(\OPT \le v(E)\le n\,v_{\max}\),
    \[
    \log v(E)\le \log(nv_{\max})=O(\log n+\log v_{\max}).
    \]
    Let \(b\) be the bit-length of \(v_{\max}\). If \(b=\poly(n)\) we are done; otherwise \(b\) is exponential in \(n\) and we can solve the instance exactly by enumerating all \(2^n\) subsets of $E$, which is still polynomial in $|E|$.
\end{proof}

We assume w.l.o.g.\ that $v_e \le 2\hat V$ for all $e\in E$ (any $e$ with
$v_e>2\hat V$ cannot belong to an optimal solution, since all values are nonnegative and
$\OPT \le 2\hat V$ by Lemma~\ref{lem:Vhat}).
We partition the ground set $E$ into bins according to the values $v_e$. We first define the bin of \emph{small} elements:
\[
  B_0 := \{ e \in E : v_e \le \eps \hat{V} \}.
\]
All other elements are considered \emph{large} and are further categorized into $T = \Bigl\lceil \frac{2-\eps}{\eps^2} \Bigr\rceil$ bins:
\[
  B_i := \bigl\{ e \in E : (\eps+(i-1)\eps^2)\hat{V} < v_e \le (\eps + i \eps^2) \hat{V} \bigr\},
  \qquad i=1,\ldots,T.
\]

\begin{lemma} \label{lem:few-large}
In any optimal solution $S^\opt$, at most $2/\eps$ elements belong to $\bigcup_{i=1}^T B_i$.
\end{lemma}

\begin{proof}
Suppose, for the sake of contradiction, that $S^\opt$ contains more than $2/\eps$ elements from $B_1\cup\cdots\cup B_T$. By the
definition of these bins, every such element has value strictly larger than $\eps \hat V$, and hence
\[
 v\bigl(S^\opt \bigr) \geq v\bigl(S^\opt \cap (B_1\cup\cdots\cup B_T)\bigr)
  \;>\; \frac{2}{\eps}\cdot \eps \hat V
  \;=\; 2\hat V.
\]
Hence $v(S^\opt)>2\hat V$. If $\opt=\max$, then by Lemma~\ref{lem:Vhat} we have $\OPT\le 2\hat V$, contradicting $v(S^\opt)=\OPT>2\hat V$. If $\opt=\min$, then Lemma~\ref{lem:Vhat} gives $\OPT\le \hat V$ and hence $\OPT\le 2\hat V$, again contradicting $v(S^\opt)=\OPT>2\hat V$.
 Thus $S^\opt$ contains at most $2/\varepsilon$ elements from $\bigcup_{i=1}^T B_i$.
\end{proof}

We therefore restrict our attention to solutions with at most $2/\eps$ large elements. We encode such solutions via a \emph{configuration vector} $\vec{k} = (k_0,k_1,\dots,k_T)$, where $k_i$ denotes the number of elements that should be chosen from bin $B_i$. We only consider vectors satisfying
\begin{equation}
\label{eq:config-constraints}
  \sum_{i=1}^T k_i \le \left\lceil \frac{2}{\eps} \right\rceil,
  \qquad
  \sum_{i=0}^T k_i \le |E|.
\end{equation}
For such a vector $\vec{k}$, a set $S \subseteq E$
\emph{respects the configuration} $\vec{k}$ if $|S \cap B_i| \;=\; k_i \text{ for all } i = 0,\dots,T.$

\begin{lemma}\label{lem:num-configs}
The number of distinct configuration vectors satisfying \eqref{eq:config-constraints} is at most $O\left( |E|\cdot \left(\frac{1}{\eps^2}\right)^{2/\eps}\right)$.
\end{lemma}

\begin{proof}
Let $m := \lceil 2/\eps \rceil$. By \eqref{eq:config-constraints}, the number of valid tuples $(k_1, \dots, k_T)$ corresponds to the number of non-negative integer solutions to the inequality $\sum_{i=1}^T k_i \le m$.
By introducing a slack variable $k^\star \ge 0$ such that $k^\star + \sum_{i=1}^T k_i = m$, this count is equivalent to the number of weak compositions of $m$ into $T+1$ parts. Applying the stars and bars formula, we obtain exactly
\[
\binom{(T+1) + m - 1}{m} = \binom{T + m}{m}
\]
configurations. For each fixed tuple, $k_0$ is bounded by $|E|$. For a fixed $\eps$ (and thus fixed $m$), the number of configurations is
\[
|E|\cdot\binom{T + m}{m}
\;\le\; |E|\cdot (T+m)^m
\;=\; O\!\left(|E|\cdot T^m\right)
\;=\; O\!\left(|E|\cdot\left(\frac{1}{\eps^2}\right)^{2/\eps}\right).
\]
\end{proof}

Let $\vec{k}^\opt$ denote the configuration vector of an optimal solution $S^\opt$. By Lemma~\ref{lem:few-large}, $\vec{k}^\opt$ satisfies \eqref{eq:config-constraints}, hence it is among the vectors considered by our algorithm.

\subsection{Partition Matroid and Restricted Family}

Fix a configuration vector $\vec{k}=(k_0,\ldots,k_T)$ satisfying~\eqref{eq:config-constraints}, and let $K:=\sum_{i=0}^T k_i$. We can enforce the cardinality of $\scm$ by considering its truncation to rank $K$, denoted $\scm_{\le K}$, whose independent sets are
$\sci_{\le K} = \IS(\scm_{\le K}) = \{S\in \sci: |S|\le K\}$.

Next, we encode the bin-wise upper bounds prescribed by $\vec{k}$ via a partition matroid. Let $B_0,\ldots,B_T$ be the bins (a partition of $E$). Define the partition matroid $\scm^P_{\vec{k}}$ on $E$ by
\[
S\in \sci^P_{\vec{k}} = \IS(\scm^P_{\vec{k}})\quad \Longleftrightarrow \quad |S\cap B_i|\le k_i \ \ \text{for all } i=0,\ldots,T.
\]
We restrict the \pmol problem to sets that satisfy both matroid constraints and have exactly $K$ elements:
\[
\scf_{\vec k}
\;:=\;
\{\, S\subseteq E:\ S\in \sci_{\le K} \cap \sci^P_{\vec k} \; \text{and}\; |S|=K\,\}.
\]
If $K>\rank(\scm)$, then $\scf_{\vec k}=\emptyset$ (since no independent set can have size $K$);
in the algorithm we simply skip such configurations.

Equivalently, $\scf_{\vec k}$ is the family of common bases of $\scm_{\le K}$ and $\scm^P_{\vec k}$ (since in each matroid, a set is a base iff it is independent and has size equal to the matroid rank, here $K$). If the original feasible family is $\scf=\IS(\scm)$, then the above definition applies for any $K\le \rank(\scm)$. If instead $\scf=\bases(\scm)$, then we must have $K=\rank(\scm)$; otherwise no set of size $K$ can be a base of $\scm$. Accordingly, in the bases variant we only consider configuration vectors with $K=\rank(\scm)$ (and for other $K$ we take $\scf_{\vec k}=\emptyset$).

Finally, note that membership in $\scf_{\vec k}$ forces the bin constraints of $\scm^P_{\vec{k}}$ to be tight. Indeed, for any $S\in\scf_{\vec k}$,
we have
\[
|S\cap B_i|\le k_i \ \ \forall i \in \{0,\ldots,T\}
\quad\text{and}\quad
\sum_{i=0}^T |S\cap B_i| = |S| = K = \sum_{i=0}^T k_i,
\]
implying that $|S\cap B_i|=k_i$ for all $i$. In particular, every $S\in\scf_{\vec k}$ is a base of the partition matroid
$\scm^P_{\vec k}$ and also a base of the truncated matroid $\scm_{\le K}$ since $|S|=K$.



For $\opt = \max$, let $\OPT_{\vec{k}}$ denote the maximum of $v(S)$ over all
$S \in \scf_{\vec{k}}$ satisfying $w(S) \tr L$; for $\opt = \min$ we define
$\OPT_{\vec{k}}$ analogously as a minimum.
If $\vec{k} = \vec{k}^\opt$ is the configuration of an optimal solution,
then clearly $\OPT_{\vec{k}^\opt} = \OPT$.

\subsection{Lagrangian Relaxation}
Returning to our problem, we must satisfy both the matroid-intersection constraint and the linear constraint: $w(S)\tr L$. Since the linear constraint is typically the harder one to handle directly, we apply \emph{Lagrangian relaxation} to incorporate it into the objective via a penalty term.

To express the different orientations in a unified form, we define two sign parameters:
\[
s_{\opt} =
\begin{cases}
+1 & \text{if } \opt = \max, \\
-1 & \text{if } \opt = \min,
\end{cases}
\qquad
s_{\tr} =
\begin{cases}
+1 & \text{if } \tr = \geq, \\
-1 & \text{if } \tr = \leq.
\end{cases}
\]
For $\lambda \ge 0$ and $S \subseteq E$, we define the Lagrangian objective
\[
  \phi_\lambda(S)
  \;:=\;
  v(S) + s_{\opt} s_{\tr} \lambda \bigl( w(S) - L \bigr).
\]
The corresponding Lagrangian relaxation value for some fixed configuration vector $\vec{k}$ is
\begin{equation}
\label{prob:lagrangian-k}
  V_{\vec{k}}(\lambda)
  \;:=\;
  \opt_{S \in \scf_{\vec{k}}} \, \phi_\lambda(S).
\end{equation}
Finally, we define the optimal value of the Lagrangian relaxation by
\[
  \LR_{\vec{k}} :=
  \begin{cases}
    \displaystyle \min_{\lambda \ge 0} V_{\vec{k}}(\lambda) & \text{if } \opt = \max,\\[2mm]
    \displaystyle \max_{\lambda \ge 0} V_{\vec{k}}(\lambda) & \text{if } \opt = \min.
  \end{cases}
\]

\begin{lemma}\label{lem:Vk-convex-concave}
For every configuration vector $\vec{k}$, the function
$V_{\vec{k}} : \mathbb{R}_{\ge0} \to \mathbb{R}$ is piecewise linear.
Moreover, if $\opt = \max$, then $V_{\vec{k}}$ is convex, and if
$\opt = \min$, then $V_{\vec{k}}$ is concave.
\end{lemma}

\begin{proof}
Fix $\vec{k}$ and consider any $S\in\scf_{\vec{k}}$. The function $\phi_\lambda(S)$ is affine in $\lambda$. Therefore, $V_{\vec{k}}(\lambda)$ is the pointwise optimum of finitely many affine functions (one for each $S\in\scf_{\vec{k}}$), and hence $V_{\vec{k}}$ is piecewise linear.

Moreover, if $\opt=\max$ then $V_{\vec{k}}$ is the pointwise maximum of affine functions and hence convex; if $\opt=\min$ then $V_{\vec{k}}$ is the pointwise minimum of affine functions and hence concave.
\end{proof}

\begin{lemma} 
\label{lem:lagrangian-bound-k}
For every configuration vector $\vec{k}$, if $\opt = \max$, then $\LR_{\vec{k}} \ge \OPT_{\vec{k}}$, and if $\opt = \min$, then $\LR_{\vec{k}} \le \OPT_{\vec{k}}$.
\end{lemma}

\begin{proof}
Let $S$ be any feasible solution with configuration $\vec{k}$, i.e.,
$S \in \scf_{\vec{k}}$ and $w(S) \tr L$.
Then $s_{\tr}\bigl(w(S)-L\bigr) \ge 0$.
For any $\lambda \ge 0$,
\[
  s_{\opt}\,\phi_\lambda(S)
  = s_{\opt}\,v(S) + s_{\tr}\,\lambda\bigl(w(S)-L\bigr)
  \;\ge\;
  s_{\opt}\,v(S).
\]
Optimizing over all such $S$ yields
$s_{\opt} V_{\vec{k}}(\lambda) \ge s_{\opt} \OPT_{\vec{k}}$.
If $\opt = \max$ (so $s_{\opt}=+1$), this implies
$V_{\vec{k}}(\lambda) \ge \OPT_{\vec{k}}$ for all $\lambda$, hence
$\LR_{\vec{k}} = \min_{\lambda \ge 0} V_{\vec{k}}(\lambda) \ge \OPT_{\vec{k}}$.
The minimization case is analogous.
\end{proof}

We now argue that the Lagrangian relaxation \eqref{prob:lagrangian-k} can be solved efficiently.

\begin{lemma}\label{lem:lagrangian-solve-intersection}
For every configuration vector $\vec{k}$ and any $\lambda \ge 0$, the value $V_{\vec{k}}(\lambda)$ and a set $S_{\vec{k},\lambda} \in \scf_{\vec{k}}$ attaining this value can be computed in polynomial time.
Moreover, a multiplier $\lambda^\star_{\vec{k}}$ satisfying
$V_{\vec{k}}(\lambda^\star_{\vec{k}}) = \LR_{\vec{k}}$ can be computed in polynomial time.
\end{lemma}

\begin{proof}
Fix a configuration vector $\vec{k}$ and $\lambda \ge 0$, and define the modified weight function $v^\lambda$ as follows
\[
  v^\lambda_e := v_e + s_{\opt}s_{\tr}\lambda w_e \quad \text{for all } e \in E.
\]
Then, for any $S \subseteq E$, we can write
\[
  \phi_\lambda(S)
  = v(S) + s_{\opt}s_{\tr}\lambda\bigl(w(S)-L\bigr)
  = \sum_{e \in S} v^\lambda_e \;-\; s_{\opt}s_{\tr}\lambda L.
\]
Hence, for fixed $\lambda$, optimizing $\phi_\lambda$ over $S \in \scf_{\vec{k}}$ is equivalent (up to an additive constant independent of $S$) to solving the following optimization problem
\[
  \opt_{S \in \scf_{\vec{k}}} \sum_{e \in S} v^\lambda_e.
\]

For both $\scf\in\{\IS,\bases\}$, the restriction to $\scf_{\vec{k}}$ imposes a fixed cardinality constraint $|S|=K$ (where $K=\sum_i k_i$, and in the bases case necessarily $K=\rank(\scm)$, otherwise $\scf_{\vec{k}}=\emptyset$). Consequently, the problem of optimizing a linear objective over $\scf_{\vec{k}}$ is a cardinality-constrained weighted matroid intersection instance between $\scm$ (or $\scm_{\leq K}$) and the partition matroid $\scm^P_{\vec{k}}$. This problem is solvable in polynomial time by the algorithm of Brezovec et al.~\cite{brezovec1986two}. Therefore, for any fixed $\vec{k}$ and $\lambda\ge 0$, we can compute $V_{\vec{k}}(\lambda)$ in polynomial time.

It remains to find a multiplier $\lambda^\star_{\vec{k}}\ge 0$ such that
$V_{\vec{k}}(\lambda^\star_{\vec{k}})=\LR_{\vec{k}}$.
Since the objective weights $v^\lambda_e$ depend linearly on $\lambda$, we can apply Megiddo's parametric-search framework~\cite{megiddo1978combinatorial} to compute such an optimal multiplier $\lambda^\star_{\vec{k}}$ in polynomial time.
\end{proof}

Let $\lambda^\star_{\vec{k}}$ be such an optimal multiplier, and define the family of
Lagrangian--optimal solutions
\[
  \sco_{\vec{k}}
  := \{\, S \in \scf_{\vec{k}} : \phi_{\lambda^\star_{\vec{k}}}(S) = \LR_{\vec{k}} \,\}.
\]
For every $S \in \sco_{\vec{k}}$ we have $\phi_{\lambda^\star_{\vec{k}}}(S) = \LR_{\vec{k}}$,
which can be rewritten as
\begin{equation}
\label{eq:v-vs-LR-k}
  v(S)=\LR_{\vec{k}} - s_{\opt} s_{\tr}\,\lambda^\star_{\vec{k}}\bigl(w(S)-L\bigr).
\end{equation}

We next consider the degenerate edge case in which $\lambda^\star_{\vec{k}}=0$. In this case, the Lagrangian objective reduces to the original objective over $\scf_{\vec{k}}$. Let
\[
S^0_{\vec{k}}\in \operatorname*{arg\;opt}_{S\in\mathcal F_{\vec{k}}} v(S)
\]
be an unconstrained optimum chosen to maximize $s_{\tr}(w(S)-L)$ among all unconstrained optima. If $S^0_{\vec{k}}$ is feasible, we add it directly to the candidate pool. If no unconstrained optimum is feasible, this configuration contributes no candidate in the case $\lambda^\star_{\vec{k}}=0$. Henceforth, we assume that $\lambda^\star_{\vec{k}}>0$.


\subsection{Spread and Adjacency}

We introduce a bound on the change in $v(\cdot)$ under a one-per-bin exchange, and then prove an adjacency property of $\sco_{\vec{k}}$ showing that any two solutions in $\sco_{\vec{k}}$ are connected by a sequence of such exchanges. For each bin $B_i$ define its \emph{spread}
\[
  \Delta_i \;:=\; \max_{e\in B_i} v_e \;-\; \min_{e\in B_i} v_e,
\]
with the convention that $\Delta_i = 0$ if $B_i = \emptyset$. For a configuration vector $\vec{k}$, define the total spread
\[
  \Sigma_\Delta(\vec{k}) \;:=\; \sum_{i:\,k_i>0} \Delta_i.
\]

\begin{lemma}\label{lem:spread-bound}
Fix a configuration vector $\vec{k}$ satisfying \eqref{eq:config-constraints}.
\begin{enumerate}
  \item\label{it:spread-step}
  If $S,S'\in\scf_{\vec{k}}$ differ by exchanging at most one element per bin, then $|v(S)-v(S')|\le \Sigma_\Delta(\vec{k})$.

  \item\label{it:spread-total}
  It holds that $\Sigma_\Delta(\vec{k}) \le 8\eps\,\OPT$.
\end{enumerate}
\end{lemma}

\begin{proof}
\emph{Part \eqref{it:spread-step}.}
Fix a bin $B_i$ with $k_i>0$. By assumption, passing from $S$ to $S'$ removes at most one element from $B_i$ and therefore (to maintain $|S'\cap B_i|=|S\cap B_i|$, since $S,S'\in\scf_{\vec{k}}$) it also adds at most one element to $B_i$. Hence the contribution of bin $B_i$ to the
value can change by at most $\Delta_i$, i.e.,
$|v(S\cap B_i)-v(S'\cap B_i)|\le \Delta_i$.
Summing over all bins with $k_i>0$ yields
\[
  |v(S)-v(S')|
  \le \sum_{i:\,k_i>0} |v(S\cap B_i)-v(S'\cap B_i)|
  \le \sum_{i:\,k_i>0} \Delta_i
  = \Sigma_\Delta(\vec{k}).
\]

\noindent
\emph{Part \eqref{it:spread-total}.}
By the bucketing construction, $\Delta_0\le \eps \hat{V}$ and $\Delta_i\le \eps^2 \hat{V}$ for all
$i\ge 1$. Moreover, \eqref{eq:config-constraints} implies that the number of indices $i\ge 1$ with $k_i>0$ is at most $\lceil 2/\eps\rceil$. Therefore,
\[
  \Sigma_\Delta(\vec{k})
  = \Delta_0\cdot \mathbf{1}[k_0>0] + \sum_{i\ge 1:\,k_i>0}\Delta_i
    \le \eps \hat{V} + \left\lceil\frac{2}{\eps}\right\rceil \cdot \eps^2 \hat{V}
  \le \eps \hat{V} + 3\eps \hat{V}
  \le 4\eps \hat{V}.
\]
Finally, by Lemma~\ref{lem:Vhat}, we have $\hat{V}\le \OPT$ when $\opt=\max$, and $\hat{V}\le 2\OPT$ when
$\opt=\min$; thus, in both cases $\Sigma_\Delta(\vec{k}) \le 8\eps\,\OPT$.
\end{proof}

We next establish the adjacency property of $\sco_{\vec{k}}$, namely that any two sets in $\sco_{\vec{k}}$ can be connected by a sequence of exchanges that swaps at most one element per bin. Our proof uses the following theorem of Brezovec et al.~\cite{brezovec1988matroid}.

\begin{lemma}[Theorem~4 from~\cite{brezovec1988matroid}]\label{claim:BCG}
Let $\scm_1=(E,\sci_1)$ be a matroid and let $\scm_2=(E,\sci_2)$ be a partition matroid on the same ground set $E$, given by a partition $E=B_0\cup\cdots\cup B_T$ and capacities $(k_0,\ldots,k_T)$. Let $c : E \to \mathbb{R}$ be a weight function, and let $S$ be a common base of $\scm_1$ and $\scm_2$.

Define the directed bipartite \emph{exchange graph} $B_S=(U\cup V,A)$ by $U := S, V := E\setminus S$, (i.e., one vertex for each element), and for each $e\in U$ and $f\in V$ define arcs as follows:
\begin{itemize}
  \item if $S-e+f \in \sci_1$, add an arc $(e,f)$ with cost $c(f)-c(e)$;
  \item if $S-e+f \in \sci_2$, add an arc $(f,e)$ with cost $0$.
\end{itemize}

Let $D_S$ be the directed graph obtained from $B_S$ by compressing, for each $i\in\{0,\ldots,T\}$, all vertices in $V\cap B_i$ into a single vertex $v_i$, and then deleting parallel arcs while keeping, for each ordered pair of vertices, only the arc of minimum length.

For the graph $D_S$ the following hold:
\begin{itemize}
    \item There are no negative cycles in $D_S$ if and only if $S$ is optimal for the weighted matroid intersection instance.

    \item If $D_S$ contains a directed cycle $C$ of negative total length, then there exists a common base $T$ of $\scm_1$ and $\scm_2$ that is obtained from $S$ by exchanging elements along $C$ and satisfies $c(T) < c(S).$
\end{itemize}
\end{lemma}

Using this lemma, we obtain the following exchange-path guarantee.
\begin{lemma}\label{lem:exchange-path}
For any $S, S' \in \sco_{\vec{k}}$, we can find in polynomial time a sequence of solutions
$S = S_0, S_1, \ldots, S_\ell = S'$ such that:
\begin{enumerate}
  \item $S_j \in \sco_{\vec{k}}$ for all $j = 0, \ldots, \ell$;
  \item For each $j = 0, \ldots, \ell - 1$ and for each bin $B_i$,
        \[
          |(S_j \setminus S_{j+1}) \cap B_i|
          =
          |(S_{j+1} \setminus S_j) \cap B_i|
          \;\leq\; 1.
        \]
\end{enumerate}
\end{lemma}

\begin{proof}[Proof of Lemma~\ref{lem:exchange-path}]
Fix a configuration vector $\vec{k}$, and let $v^* = v^{\lambda^*_{\vec{k}}}$. W.l.o.g., assume that $\sco_{\vec{k}}$ consists of \emph{minimum}-cost bases with respect to $v^\ast$ (otherwise, replace $v^\ast$ by $-v^\ast$).

We prove the statement by induction on $x := |S' \setminus S|$. The case $x=0$ is trivial. Assume that $x\ge 1$ and the lemma holds for every pair $T,T' \in \sco_{\vec{k}}$ satisfying $|T'\setminus T|<x$. We prove the claim for the given pair $S,S'$ with $|S'\setminus S|=x$.

Fix $\delta>0$ and define $v' : E \to \mathbb{R}\cup\{+\infty\}$ by
\[
  v'(e)=
  \begin{cases}
    v^\ast(e)+\delta, & e\in S\setminus S',\\
    v^\ast(e)-\delta, & e\in S'\setminus S,\\
    v^\ast(e),             & e\in S\cap S',\\
    +\infty,               & e\notin S\cup S'.
  \end{cases}
\]
Then $v'(S)=v^\ast(S)+x\delta$ and $v'(S')=v^\ast(S')-x\delta$, and hence $v'(S') = v^\ast(S')-x\delta < v^\ast(S)+x\delta = v'(S)$, where we used $v^\ast(S)=v^\ast(S')$ since $S,S'\in\sco_{\vec{k}}$. To avoid choosing an explicit infinitesimal $\delta$, we implement this perturbation lexicographically. Let $U$ be an integer upper bound on
$\sum_{e\in S}|v^\ast(e)|$ over all $S\in\scf_{\vec k}$ (e.g.,\ $U:=K\cdot\max_{e\in E}|v^\ast(e)|$),
and set $M:=2U+1$.
Define an auxiliary weight function $\tilde v:E\to\mathbb R$ by
\[
\tilde v(e)=
\begin{cases}
M\cdot v^\ast(e)+1, & e\in S\setminus S',\\
M\cdot v^\ast(e)-1, & e\in S'\setminus S,\\
M\cdot v^\ast(e),   & e\in S\cap S',\\
M\cdot v^\ast(e)+M^2+K, & e\notin S\cup S'.
\end{cases}
\] 
We claim that $\tilde v(S') < \tilde v(S)$. Indeed, since $S,S'\in\sco_{\vec k}$ we have
$v^*(S)=v^*(S')$. Moreover, letting $x:=|S'\setminus S|=|S\setminus S'|$, the definition of $\tilde v$
implies
\[
\tilde v(S)=\sum_{e\in S}\tilde v(e)
= M\sum_{e\in S} v^*(e) + |S\setminus S'|
= M\,v^*(S) + x,
\]
and similarly
\[
\tilde v(S')=\sum_{e\in S'}\tilde v(e)
= M\sum_{e\in S'} v^*(e) - |S'\setminus S|
= M\,v^*(S') - x.
\]
Therefore,
\[
\tilde v(S)-\tilde v(S')
= M\bigl(v^*(S)-v^*(S')\bigr) + 2x
= 2x > 0,
\]
and hence $\tilde v(S')<\tilde v(S)$.

\medskip
Since the weights $v^\ast(e)$ may be rational, we first scale them to integers.
Let $D$ be a common denominator of $\{v^\ast(e) : e\in E\}$, and define
$\bar v^\ast(e):=D\cdot v^\ast(e)$ for all $e\in E$.
Then $\bar v^\ast(e)\in \mathbb Z$ for all $e$, and for any two common bases $T_1,T_2$
we have $\bar v^\ast(T_1)-\bar v^\ast(T_2)\in \mathbb Z$.

We redefine $\tilde v$ using $\bar v^\ast$ as follows.
Let $U:=K\cdot \max_{e\in E}|\bar v^\ast(e)|$ and set $M:=2U+1$.
Define $\tilde v:E\to \mathbb Z$ by
$\tilde v(e)=M\cdot \bar v^\ast(e)+1$ for $e\in S\setminus S'$,
$\tilde v(e)=M\cdot \bar v^\ast(e)-1$ for $e\in S'\setminus S$,
$\tilde v(e)=M\cdot \bar v^\ast(e)$ for $e\in S\cap S'$,
and 
$\tilde v(e)=M\cdot \bar v^\ast(e)+M^2+2MU+2K$ for $e\notin S\cup S'$.

Our redefinition preserves the computation for $S$ and $S'$.
In particular, using $\bar v^\ast(S)=\bar v^\ast(S')$ we have 
$\tilde v(S)=M\cdot \bar v^\ast(S)+x$ and $\tilde v(S')=M\cdot \bar v^\ast(S')-x$,
so $\tilde v(S')<\tilde v(S)$.

\begin{claim}
\label{claim:min_common_base}
Any $\tilde v$-minimum common base is (i) a $v^\ast$-minimum common base (and hence belongs to $\sco_{\vec k}$), and (ii) consists solely of elements from $S\cup S'$.
\end{claim}

\noindent	{\em Proof.}
We first show that minimizing $\tilde v$ enforces minimization of $v^\ast$ (lexicographically).
Fix any common base $T\in\scf_{\vec k}$. 
Write $\tilde v(T)= M\cdot \bar v^\ast(T) + \tau(T)$,
where $\tau(T)$ is the additive tie-breaking term induced by
the $\pm 1$ offsets (and $0$ offsets) in the definition of $\tilde v$.
Since $|T|=K$, we have the uniform bound $|\tau(T)|\le K$ whenever $T\subseteq S\cup S'$.
Moreover, if $T$ contains some element $e\notin S\cup S'$, then
\[
\tilde v(T)
\ge M\cdot \bar v^\ast(T) - K + (M^2+2MU+2K)
\ge -MU - K + (M^2+2MU+2K)
= M^2+MU+K.
\]
On the other hand, for every common base $T\subseteq S\cup S'$ we have
\[
\tilde v(T)=M\cdot \bar v^\ast(T)+\tau(T)\le MU+K.
\]
Therefore any $\tilde v$-minimum common base must be contained in $S\cup S'$.

Now let $T_1,T_2\subseteq S\cup S'$ be any two common bases. Then
\[
\tilde v(T_1)-\tilde v(T_2)
= M\bigl(\bar v^\ast(T_1)-\bar v^\ast(T_2)\bigr) + \bigl(\tau(T_1)-\tau(T_2)\bigr),
\]
and $|\tau(T_1)-\tau(T_2)|\le 2K$.
Choose $U$ so that $U\ge K\cdot \max_{e\in E}|\bar v^\ast(e)|$, and set $M:=2U+1$ as above.
Then $M>2K$. Moreover, since $\bar v^\ast(T_1),\bar v^\ast(T_2)\in\mathbb Z$, if
$\bar v^\ast(T_1)>\bar v^\ast(T_2)$ then $\bar v^\ast(T_1)-\bar v^\ast(T_2)\ge 1$, and hence
\[
\tilde v(T_1)-\tilde v(T_2)
= M\bigl(\bar v^\ast(T_1)-\bar v^\ast(T_2)\bigr)+\bigl(\tau(T_1)-\tau(T_2)\bigr)
\ge M-2K>0.
\]
Therefore $\tilde v(T_1)>\tilde v(T_2)$.
Consequently, among all common bases contained in $S\cup S'$, every $\tilde v$-minimum base is also
$v^\ast$-minimum. 
Moreover, as shown above, if a common base $T$ contains an element outside $S\cup S'$, then
$\tilde v(T)\ge M^2+MU+K$, whereas every common base $T\subseteq S\cup S'$ satisfies $\tilde v(T)\le MU+K$.
Hence no $\tilde v$-minimum common base uses elements outside $S\cup S'$.
Therefore, any $\tilde v$-minimum common base belongs to $\sco_{\vec k}$ and is contained in $S\cup S'$.
\hfill $\square$

\vspace{1mm}

Apply Lemma~\ref{claim:BCG} to the pair $(\scm_{\le K},\scm^P_{\vec{k}})$, the weight function $\tilde v$, and the current base $S$. Since $\tilde v(S')<\tilde v(S)$, the base $S$ is not $\tilde v$-optimal, so the compressed graph $D_S$ contains a negative cycle. Let $C$ be a negative directed cycle in $D_S$
with the minimum number of arcs. Lemma~\ref{claim:BCG} yields a common base $S''$ obtained from $S$ by exchanging along $C$ such that $\tilde v(S'')<\tilde v(S)$. 
By the choice of $\tilde v(e)$ for $e\notin S\cup S'$, we have $S''\subseteq S\cup S'$.

Since $D_S$ is bipartite, any directed cycle $C$ alternates between vertices $e\in S$ and contracted bin vertices $v_i$. We chose $C$ with the minimum number of arcs, hence $C$ cannot visit the same $v_i$ twice; otherwise $C$ would split into two directed cycles, and at least one of them would still have negative total length but fewer arcs, contradicting minimality. Thus, $C$ visits each bin at most once, so the corresponding exchange inserts at most one element from each bin $B_i$, i.e.\ $|(S''\setminus S)\cap B_i|\le 1$. Finally, $S$ and $S''$ are both bases of the partition matroid $\scm^P_{\vec k}$, so in each bin $B_i$ an insertion must be matched by a deletion from the same bin in order to keep $|S\cap B_i|=k_i$. Hence, the numbers of inserted and deleted elements match in every bin:
\[
  |(S\setminus S'')\cap B_i|
  =
  |(S''\setminus S)\cap B_i|
  \le 1
  \qquad \forall i \in \{0,\ldots,T\}.
\]

Since $S\in\sco_{\vec k}$ is $v^\ast$-minimum, we have $v^\ast(T)\ge v^\ast(S)$ for every common base $T$.
In particular, for every common base $T\subseteq S\cup S'$ we have $\bar v^\ast(T)\ge \bar v^\ast(S)$.
Now $S''\subseteq S\cup S'$ and $\tilde v(S'')<\tilde v(S)$.
If $\bar v^\ast(S'')>\bar v^\ast(S)$ then $\tilde v(S'')-\tilde v(S)\ge M-2K>0$, a contradiction.
Hence $\bar v^\ast(S'')=\bar v^\ast(S)$, i.e., $v^\ast(S'')=v^\ast(S)$, and therefore $S''\in\sco_{\vec k}$.

Finally, since the inserted elements come from $E\setminus S$ and, by the definition of $\tilde v$, must lie in $S'\setminus S$, we add at least one element of $S'\setminus S$; thus $|S'\setminus S''|<|S'\setminus S|=x$.

The induction hypothesis applies to the pair $(S'',S')$. Hence there exists a sequence $S''=S_0,S_1,\ldots,S_m=S'$ with $S_j\in\sco_{\vec{k}}$ for all $j$ and such that each transition $S_j\to S_{j+1}$ exchanges at most one element per bin. Appending the initial step $S\to S''$ to this sequence yields the desired path from $S$ to $S'$.

The above inductive argument yields an explicit polynomial-time procedure to construct the sequence.
Given $S,S'\in\sco_{\vec k}$, form the perturbed weight function $\tilde v$ as in the proof.
Since $\tilde v(S')<\tilde v(S)$, the base $S$ is not $\tilde v$-optimal, so by Lemma~\ref{claim:BCG}
the compressed exchange graph $D_S$ contains a negative directed cycle, which can be found in polynomial time (e.g.\ by Bellman--Ford; extracting a \emph{simple} negative cycle ensures no bin-vertex $v_i$ is repeated). Exchanging along this cycle produces $S''\in\sco_{\vec k}$ that differs from $S$ by at most one element per bin and satisfies $|S'\setminus S''|<|S'\setminus S|$. Repeating for at most $|S'\setminus S|\le K$ iterations constructs $S=S_0,S_1,\ldots,S_\ell=S'$.
\end{proof}

\subsection{Deriving an EPTAS}

We now show that for every feasible configuration vector $\vec{k}$, there is a Lagrangian–optimal solution that is both feasible for the original linear constraint and nearly optimal in value.

\begin{lemma}\label{lem:structural}
    Fix a configuration vector $\vec{k}$ with $\scf_{\vec{k}} \neq \emptyset$, and
    let $\lambda^\star_{\vec{k}}$, $\LR_{\vec{k}}$, and $\sco_{\vec{k}}$ be as defined above.
    Then, there exists a set $S^\star \in \sco_{\vec{k}}$ such that $w(S^\star) \tr L$ and the following holds:
    \begin{itemize}
      \item if $\opt = \min$, then $v(S^\star) \le \LR_{\vec{k}} + \Sigma_\Delta(\vec{k})$,
      \item if $\opt = \max$, then $v(S^\star) \ge \LR_{\vec{k}} - \Sigma_\Delta(\vec{k})$.
    \end{itemize}
\end{lemma}

\begin{proof}
    Recall that $\sco_{\vec{k}}=\arg\opt_{S\in\scf_{\vec{k}}}\phi_{\lambda_{\vec{k}}^\star}(S)$.
    Choose $S^-,S^+\in\sco_{\vec{k}}$ such that
    $s_{\tr}(w(S^-)-L)$ is minimized over $\sco_{\vec{k}}$ and
    $s_{\tr}(w(S^+)-L)$ is maximized over $\sco_{\vec{k}}$
    (These can be computed in polynomial time by lexicographic optimization over $\scf_{\vec{k}}$:
    first optimize $\phi_{\lambda_{\vec{k}}^\star}$, and among ties minimize/maximize $s_{\tr}w(S)$, implemented via
    weighted matroid intersection).
    
    We claim that $s_{\tr}(w(S^-)-L)\le 0 \le s_{\tr}(w(S^+)-L)$.
    Otherwise, if $s_{\tr}(w(S)-L)>0$ for all $S\in\sco_{\vec{k}}$ (or $<0$ for all such $S$),
    then perturbing $\lambda_{\vec{k}}^\star$ slightly in the appropriate direction would improve $V_{\vec{k}}(\lambda)$,
    contradicting the optimality of $\lambda_{\vec{k}}^\star$ in the definition of $\LR_{\vec{k}}$.
    
    For simplicity, we present the remainder of the proof for the minimization case; the maximization case is analogous.
    Using Eq.~\eqref{eq:v-vs-LR-k} with $\opt = \min$ (so $s_{\opt}=-1$), we have $v(S) - \LR_{\vec{k}} = s_{\tr} \lambda^\star_{\vec{k}} \bigl(w(S)-L\bigr)$.
    Since $\lambda^\star_{\vec{k}} \ge 0$, it follows that $s_{\tr}\bigl(w(S^-)-L\bigr) \le 0$ implies $v(S^-) \le \LR_{\vec{k}}$, and $s_{\tr}\bigl(w(S^+)-L\bigr) \ge 0$ implies $v(S^+) \ge \LR_{\vec{k}}$.
    
    By Lemma~\ref{lem:exchange-path}, there is a sequence $S^- = S_0,S_1,\dots,S_\ell = S^+$ inside $\sco_{\vec{k}}$ such that each step exchanges at most one element per bin. Since $v(S_0) \le \LR_{\vec{k}}$ and $v(S_\ell) \ge \LR_{\vec{k}}$, there exists an index $j$ such that $v(S_j) \le \LR_{\vec{k}} \le v(S_{j+1})$.
    
    Combining this inequality with Lemma~\ref{lem:spread-bound} yields
    \[
      v(S_{j+1}) \le v(S_j) + \Sigma_\Delta(\vec{k}) \le
      \LR_{\vec{k}} + \Sigma_\Delta(\vec{k}).
    \]
    Moreover, since $v(S_{j+1}) \ge \LR_{\vec{k}}$ and $v(S_{j+1}) - \LR_{\vec{k}} = s_{\tr} \lambda^\star_{\vec{k}} \bigl(w(S_{j+1})-L \bigr)$ with $\lambda^\star_{\vec{k}} \ge 0$, it follows that $s_{\tr}\bigl(w(S_{j+1})-L\bigr) \ge 0$, i.e., $w(S_{j+1}) \tr L$.
    Thus $S^\star := S_{j+1}$ is feasible and satisfies the desired bound, which completes the proof.
    
    All objects used above can be computed in polynomial time. First, $\lambda^\star_{\vec k}$ is computed as in
    Lemma~\ref{lem:lagrangian-solve-intersection}, and $S^-,S^+$ are obtained by two lexicographic
    weighted-matroid-intersection calls over $\scf_{\vec k}$ (optimize $\phi_{\lambda_{\vec k}^\star}$, and among ties
    minimize/maximize $s_{\tr}w(S)$). Next, Lemma~\ref{lem:exchange-path} yields an explicit sequence
    $S^-=S_0,\ldots,S_\ell=S^+$ inside $\sco_{\vec k}$ with one-per-bin exchanges; scanning this sequence finds
    an index $j$ with $v(S_j)\le \LR_{\vec k}\le v(S_{j+1})$, and we output $S^\star:=S_{j+1}$.
\end{proof}

\begin{algorithm}[hbt]
  \caption{EPTAS for \pmol}
  \label{alg:eptas-pmol}
  \begin{algorithmic}[1]
    \Statex \textbf{Input:} An instance $\sca_\scp = (E,\sci,v,w,L)$, a descriptor $\scp = (\opt,\scf,\triangleright)$ and $\eps > 0$.
    \Statex \textbf{Output:} A feasible set $S$ which is $(1\pm\eps)$-approximation of $\OPT$.
    \State Compute a feasible set $\hat S$ and set $\hat V := v(\hat S)$, as guaranteed by Lemma~\ref{lem:Vhat}.
    \State Define the buckets $B_0,\dots,B_T \subseteq E$.
    \State $\mathcal{K} \gets \bigl\{\vec k=(k_0,\dots,k_T)\in\mathbb{Z}_{\ge0}^{T+1} :
           \sum_{i=1}^T k_i \le \lceil 2/\eps\rceil,\;
           \sum_{i=0}^T k_i \le n \bigr\}$.
    \State $S^{\mathrm{best}} \gets \hat S$
    \For{$\vec k \in \mathcal{K}$}
        \State Construct the partition matroid $\scm^P_{\vec k}$ from $\vec k$ and define $\scf_{\vec k} \subseteq 2^E$.
        \State Compute $\lambda^\star_{\vec{k}}$ and $\LR_{\vec{k}}$ (Lemma~\ref{lem:lagrangian-solve-intersection}), and compute two sets $S^-_{\vec{k}},S^+_{\vec{k}}\in\sco_{\vec{k}}$ as in Lemma~\ref{lem:structural} (with $s_{\tr}(w(S^-_{\vec{k}})-L)\le 0 \le s_{\tr}(w(S^+_{\vec{k}})-L)$).
        
        \State Using Lemmas~\ref{lem:spread-bound},~\ref{lem:exchange-path} and~\ref{lem:structural},
       compute $S^\star_{\vec{k}} \in \scf_{\vec{k}}$ with $w(S^\star_{\vec{k}}) \triangleright L$ and
       $s_{\opt}\bigl( \LR_{\vec{k}} - v(S^\star_{\vec{k}}) \bigr) \le \Sigma_\Delta(\vec{k})$.
      \If{$\opt = \min$ and $v(S^\star_{\vec k}) < v(S^{\mathrm{best}})$}
        \State $S^{\mathrm{best}} \gets S^\star_{\vec k}$
      \EndIf
      \If{$\opt = \max$ and $v(S^\star_{\vec k}) > v(S^{\mathrm{best}})$}
        \State $S^{\mathrm{best}} \gets S^\star_{\vec k}$
      \EndIf
    \EndFor
    \State \Return $S^{\mathrm{best}}$
  \end{algorithmic}
\end{algorithm}

Our main result is the following.

\begin{theorem}
\label{thm:ptas-pmol}
For every fixed $\eps > 0$, there exists an EPTAS for the \pmol problem.
Given an instance $\sca_\scp$, the algorithm runs in time $|E|^{O(1)} \cdot (1/\eps^2)^{O(1/\eps)}$ and outputs a feasible solution $S$ satisfying
\[
  \opt = \min \;\Longrightarrow\; v(S) \le (1+\eps)\,\OPT,
  \qquad
  \opt = \max \;\Longrightarrow\; v(S) \ge (1-\eps)\,\OPT.
\]
\end{theorem}

\begin{proof}
We first describe our EPTAS for the \pmol class (see \cref{alg:eptas-pmol}).
Using a constant-factor estimate $\hat V$ of $\OPT$, we first bucket the elements by their $v$-values so that elements within the same bucket have nearly equal value. We then enumerate all feasible configuration vectors $\vec k$; each such $\vec k$ defines a partition matroid $\scm^P_{\vec k}$ and a restricted family $\scf_{\vec k}$. For every configuration, we solve the Lagrangian relaxation of the linear constraint over $\scf_{\vec k}$ to obtain $\lambda^\star_{\vec k}$, $\LR_{\vec k}$ and the set $\sco_{\vec k}$ of Lagrangian-optimal solutions. Finally, using the adjacency property of $\sco_{\vec k}$ together with a bound on the per-bucket value spread, we find a set $S^\star_{\vec k}\in\scf_{\vec k}$ that satisfies the original linear constraint and is within $\Sigma_\Delta(\vec k)$ of the Lagrangian bound. The algorithm returns the best such solution over all configurations.

We now proceed to prove the approximation guarantee of \cref{alg:eptas-pmol}.
Let $S^\opt$ be an optimal solution and $\vec{k}^\opt$ its configuration vector. Lemma~\ref{lem:few-large} establishes that $\vec{k}^\opt$ satisfies Eq.~\eqref{eq:config-constraints}, so the algorithm considers this vector.
In that iteration, we get $\OPT_{\vec{k}^\opt} = \OPT$. It holds that
\[
  \opt = \min \;\Rightarrow\;
    v(S^\star_{\vec{k}^\opt})
    \le \LR_{\vec{k}^\opt} + \Sigma_\Delta(\vec{k}^\opt)
    \le \OPT + \Sigma_\Delta(\vec{k}^\opt) \le \OPT + 8\eps\OPT = (1+8\eps)\OPT,
\]
\[
  \opt = \max \;\Rightarrow\;
    v(S^\star_{\vec{k}^\opt})
    \ge \LR_{\vec{k}^\opt} - \Sigma_\Delta(\vec{k}^\opt)
    \ge \OPT - \Sigma_\Delta(\vec{k}^\opt) \ge \OPT - 8\eps\OPT = (1-8\eps)\OPT,
\]
where the first inequality in each row is by Lemma~\ref{lem:structural}, the second is due to Lemma~\ref{lem:lagrangian-bound-k} and the last holds by Lemma~\ref{lem:spread-bound}. Therefore, running the algorithm with parameter $\eps/8$ internally, we obtain the claimed $(1\pm\eps)$ guarantees. The feasibility of $S^\star_{\vec{k}^\opt}$ is immediate due to Lemma~\ref{lem:structural}.

The running time is dominated by the enumeration of configuration vectors (Lemma~\ref{lem:num-configs}) and, for each $\vec{k}$, a polynomial number of calls to (cardinality-constrained) weighted matroid intersection and local exchange operations.
This yields a total running time of $|E|^{O(1)} \cdot (1/\eps^2)^{O(1/\eps)}$, as required.
\end{proof}

\section{Hardness Proof for MIC}
\label{sec:hardness}
		
\begin{theorem}
\label{thm:MIC}
    The \textnormal{\sc Matroid Intersection Cover} problem admits no \textnormal{EPTAS}, unless $W[1]$ = \textnormal{FPT}.
\end{theorem}
	
\noindent	{\em Proof.}
We use a reduction from the $k$-clique problem. Given a graph $G = (V,E)$ and $k \in \mathbb{N}$ we need to decide if there is a clique of size $k$ in $G$. 
Let $n = |V|$, and denote by $C(G)$ the set of all cliques in $G$. We define below a MIC instance, where we assume, without loss of generality, that $0<k \leq n$ and $n \geq 1$; otherwise, the problem becomes trivial. 

Let $V = \{v_1, \ldots, v_n\}$, and for each vertex $v_i \in V$ define a copy $u_i$, forming the set $U = \{u_1, \ldots, u_n\}$, where $U \cap V = \emptyset$.  
We refer to $u_i$ and $v_i$ as the two \emph{copies} of the same vertex.  
Define a function $f : V \cup U \to [n]$ satisfying the following properties:
\begin{enumerate}
    \item The restriction of $f$ to $V$ is a bijection onto $[n]$.
    \item For every $i \in [n]$, it holds that $f(u_i) = f(v_i)$.
\end{enumerate}

For a set $T\subseteq V\cup U$, we write $f(T):= \left\{ f(x) \mid x\in T \right\}$ for the image of $T$ under $f$. We also write $\sigma_f(T):= \sum_{x\in T} f(x)$ for the sum of the labels of the elements in $T$. Using the function $f$, we define a matroid $M_1 = (V \cup U, \scf_1)$, where the independent sets are given by $\scf_1 = \sci_1 \cup \sci_2 \cup \sci_3 \cup \sci_4$, with the four families defined as follows: \begin{equation}
\label{eq:independentSetsDef}
\begin{aligned}
    \sci_1 = {} & \{T \subseteq V \cup U~\big|~ |T| < n\} \\
    \sci_2 = {} & \{T \subseteq V \cup U~\big|~ |T| = n, f(T) \neq f(V)\} \\
    \sci_3 = {} & \{T \subseteq V \cup U~\big|~ |T| = n, |T \cap V| \neq k\}. \\
    \sci_4 = {} & \{T \subseteq V \cup U~\big|~ |T| = n, |T \cap V| = k, f(T) = f(V), T \cap V \in C(G)\}. \\
\end{aligned}
\end{equation}
		
\begin{claim}
\label{claim:M1IsMatroid}
    $M_1$ is a matroid.
\end{claim}

\begin{proof}
    We verify the three matroid axioms. First, $\emptyset \in \sci_1$ since $0 < n$, and therefore $\emptyset \in \scf_1$. Second, for the hereditary property, let $A \in \scf_1$ and let $B \subseteq A$. Since $|B| < |A| \leq n$, we have $|B| < n$, hence $B \in \sci_1 \subseteq \scf_1$. Third, for the exchange property, let $A,B \in \scf_1$ such that $|A| > |B|$. We distinguish cases based on the size of $B$.
    \begin{itemize}
        \item $|B| < n-1$. Then, for all $e \in A \setminus B$ it holds that $|B+e| < (n-1)+1 = n$, and hence $B+e \in \sci_1 \subseteq \scf_1$. Such an element $e \in A \setminus B$ exists because $|A| > |B|$.

        \item $|B| = n-1$ and $|A| = n$. We further distinguish two subcases.
        \begin{enumerate}
            \item $B \subseteq A$. Then, $A \setminus B = \{e\}$ for some $e$, and thus $B+e = A \in \scf_1$.
                        
            \item $B \not\subseteq A$. Since $|B| = n-1$ and $|A| = n$, we have $|B \cap A| \leq |B|-1 = n-2$. Hence, $|A \setminus B| = |A| - |A \cap B| \geq n - (n-2) = 2$, so there exist $x,y \in A \setminus B$, such that $x \neq y$. We now consider two cases according to their images under $f$.
            \begin{enumerate}
                \item $f(x) \neq f(y)$. Therefore, there is $q \in \{x,y\}$ such that $f(B+q) \neq f(V)$. We conclude that $B+q \in \sci_2 \subseteq \scf_1$, as required.
                
                \item $f(x) = f(y)$. By the definition of $f$, there exists $i \in [n]$ such that $\{x,y\} = \{v_i, u_i\}$ (or vice versa). Hence, adding $x$ or $y$ to $B$ changes the number of selected vertices from $V$, i.e., $|(B+x) \cap V| \neq |(B+y) \cap V|$. Thus, it follows that for at least one $q \in \{x,y\}$, $|(B+q) \cap V| \neq k$, implying $B+q \in \sci_3 \subseteq \scf_1$.
            \end{enumerate}
        \end{enumerate}
    \end{itemize}
\end{proof}

We now define a second matroid, $M_2 = (U \cup V, \scf_2)$, by
\begin{equation}
	\scf_2 = \{T \subseteq V \cup U~|~ |T \cap \{u_i, v_i\}| \leq 1~\forall i \in [n]\}.
\end{equation}
Note that $M_2$ is a partition matroid, in which any independent set contains at most one copy of each vertex.
We define a weight function $w : (U \cup V) \to \drz$ by
\begin{equation*}
	w(x) = 
	\begin{cases}
		1+\frac{f(x)}{2n^2},&  x \in V, \\
		\frac{f(x)}{2n^2} &  x \in U. 
	\end{cases}
\end{equation*}
and a valuation function $v : (U \cup V) \to \drz$, given by
\begin{equation*}
	v(x) = 
	\begin{cases}
		1,&  x \in V, \\
		0 &  x \in U. 
	\end{cases}
\end{equation*}
Finally, we set the threshold $L = k+\frac{\sigma_f(V)}{2n^2}$ and define the MIC instance $\sca = (V\cup U,\sci,v,w,L)$, where $\sci = \scf_1 \cap \scf_2$.

\begin{claim}
\label{claim:cliqueIff}
    $G$ contains a clique of size $k$ if and only if $\sca$ admits a feasible solution of value $k$. 
\end{claim}

\begin{proof}
	Let $Q_V$ be a clique of size $k$ in $G$. Define $Q_U = \{u_i~|~i \in [n], v_i \notin Q_V\}$ and $K = Q_V \cup Q_U$. By construction, $K$ contains exactly one copy of each vertex, either $v_i$ or $u_i$. Therefore, $K \cap V = Q_V \in C(G)$, $|K \cap V| = |Q_V| = k$, and $f(K) = f(V)$, since $f(u_i) = f(v_i)$ for every $i \in [n]$ (and $K$ contains exactly one copy from each vertex, either from $U$ or from $V$).
    
    Thus, by \eqref{eq:independentSetsDef} we have that $K \in \sci_4$ implying that $K \in \scf_1$. Moreover, since $K$ includes exactly one element from each pair $\{u_i,v_i\}$, it also satisfies the partition constraint and hence $K \in \scf_2$. We therefore have $K \in \scf_1 \cap \scf_2$. In addition, the weight of $K$ satisfies
    \begin{equation*}
        w(K) = w(Q_V) + w(Q_U) = k+\frac{\sigma_f(Q_V)}{2n^2}+\frac{\sigma_f(Q_U)}{2n^2} = k+\frac{\sigma_f(V)}{2n^2} = L,
    \end{equation*}
    where we use that $|Q_V| = k$, $f(u_i)=f(v_i)$ for all $i \in [n]$, and $K$ contains exactly one copy of each vertex. Hence, $K$ is a feasible solution for $\sca$ with value $v(K) = |K \cap V| = |Q_V| = k$. 
	
	For the other direction, let $Q$ be a solution for $\sca$ with a value $k$. By the definition of the valuation function $v$, it follows that $|Q \cap V| = k$.
	In addition, the total weight of $Q$ is bounded by 
	\begin{equation}
		\label{eq:3}
		w(Q) = k+\frac{\sigma_f(Q)}{2n^2} \leq k+\frac{\sigma_f(V)}{2n^2} = L. 
	\end{equation}
    The first equality follows from $|Q \cap V| = k$, and the inequality holds because $Q \in \scf_2$, so $Q$ contains at most one of $\{v_i, u_i\}$ for each $i \in [n]$, so $f$ contributes at most once per vertex of $V$. Since $Q$ is feasible for $\sca$, it must also satisfy $w(Q) \ge L$, and therefore $w(Q) = L$. Combined with \eqref{eq:3}, this implies $\sigma_f(Q) = \sigma_f(V)$.
    
    Now, notice that if $|Q| < n$, then $\sigma_f(Q) < \sigma_f(V)$, since $M_2$ ensures that each vertex contributes at most once. Thus we must have $|Q| = n$. We have now established that 
    \[
    |Q \cap V| = k, \qquad |Q| = n, \qquad \sigma_f(Q) = \sigma_f(V).
    \]
    Since \(Q \in \mathcal{F}_2\), it contains at most one element from each pair \(\{u_i,v_i\}\), and therefore the labels in \(f(Q)\) are distinct. As \(|Q|=n\) and \(\sigma_f(Q)=\sigma_f(V)=1+\cdots+n\), it follows that \(f(Q)=f(V)\). Since \(Q\) is feasible, it belongs to one of the families defining \(\mathcal{F}_1\); the only possible family is $\sci_4$. By \eqref{eq:independentSetsDef}, we conclude that \(Q\cap V\in C(G)\). Hence, \(Q\cap V\) is a clique of size \(k\) in \(G\).
\end{proof}

To complete the proof, assume toward a contradiction that there exists an EPTAS for MIC. Using this scheme, we can decide whether $G$ contains a clique of size $k$ in FPT time. Consider the instance $\sca$ constructed above, and run the EPTAS with $\eps = \frac{1}{2k}$. Observe that any feasible solution for $\sca$ must include at least $k$ vertices from $V$. The weight constraint encoded by $w$ and $L$ enforces this requirement, as each vertex from $V$ contributes $1$ to the value function $v$.

By Claim~\ref{claim:cliqueIff}, $G$ has a clique of size $k$ if and only if $\sca$ admits a feasible solution of value exactly $k$. If such a solution for $\sca$ exists, then the EPTAS would yield a solution whose value is at most $(1+\eps)k = k + \frac{1}{2} < k+1$. Since the valuation function $v$ is integral, the only integer value below $k+1$ is $k$, and therefore the scheme must return a solution of value exactly $k$. Thus, the EPTAS correctly determines whether $\sca$ has a solution of value $k$ and, by Claim~\ref{claim:cliqueIff}, whether $G$ contains a clique of size $k$.

The running time of the EPTAS is $g(\frac{1}{\eps}) \cdot m^{O(1)}$ for some computable function $g$, where $m$ is the encoding size of the $k$-clique instance. Since $\eps = \Theta(\frac{1}{k})$, our reduction decides the $k$-clique problem in time $g(k)\cdot m^{O(1)}$.
As the $k$-clique problem is known to be W[1]-Hard \cite{downey1995fixed}, it implies that there is no EPTAS for MIC unless $W[1]=\text{FPT}$. \hfill$\square$

\section{Discussion}

In this paper, we present tight approximation results for a broad class of matroid optimization problems with a linear constraint, along with a hardness result for MIC. Our unified EPTAS for all \pmol problems resolves the complexity status of this class,
and the hardness of MIC suggests that such a unified EPTAS is unlikely to extend to \emph{all} \pmol variants once the single matroid constraint is replaced by a \emph{matroid intersection} constraint. 



A natural direction for future work is to characterize the approximability landscape of \pmol variants under matroid intersection constraints. In particular, the existence of a PTAS for MIC remains open.

Finally, it would be interesting to identify settings in which stronger guarantees are attainable. For instance, for graphic matroids and for linear matroids given by an explicit representation, we obtain an EPTAS, while the existence of an FPTAS in these cases remains open.


%
%
%
%
%
%
%

\appendix
\bibliographystyle{plain}
\bibliography{bibliography}

\end{document}